\documentclass[12pt]{article}    % Specifies the document style.
\usepackage{setspace}
\usepackage{hyperref}
\usepackage{graphics}
\usepackage{epsfig}

\usepackage{amsmath}
\usepackage{mathrsfs}

\usepackage{enumitem}

\usepackage{amssymb}
\usepackage{setspace}
\usepackage{hyperref}
\usepackage{graphics}
\usepackage{epsfig}
\usepackage{url}
\usepackage{amsmath}
\usepackage{bm}
\usepackage{amsfonts}
\usepackage{amsmath}
\usepackage{mathrsfs}
\usepackage{multirow}
\numberwithin{equation}{section}
\usepackage[english]{babel} % To obtain English text with the blindtext package
\usepackage{blindtext}
\usepackage{color}
\usepackage{hyperref}
\usepackage{graphicx}
\usepackage{graphics}

\usepackage{titlesec}
\usepackage{mathtools}
\usepackage{amsfonts}

\usepackage{caption}
\usepackage{subcaption}
\numberwithin{equation}{section}

\def \R {{\rm I\kern -2.2pt R\hskip 1pt}}

\newcommand{\qed}{\rule{0.5em}{1.5ex}}

\newtheorem{theorem}{Theorem} 
\newtheorem{corollary}{Corollary} 
\newenvironment{proof}{\noindent {\bf Proof.\nopagebreak}}{~\qed}

\begin{document}
	
	\begin{center} { \large \sc  Stationary Proportional Hazard Processes via Complementary Power Function Distribution Processes}\\
		\vskip 0.1in {\bf Barry C. Arnold}\footnote{barnold@ucr.edu}
		\\
		\vskip 0.1in 
		Department of Statistics, University of California, Riverside, USA. \\
		
		\vskip 0.1in {\bf  B.G. Manjunath \footnote{manjunath.raju@rgu.ac.in}}
		\\
		\vskip 0.1in 
		Department of Statistics, Rajiv Gandhi University, Arunachal Pradesh, India. \\
		
		\vskip 0.1in {\bf  S. Sachdeva  \footnote{sachinstats007@gmail.com}}
		\\
		\vskip 0.1in 
		School of Mathematics and Statistics, University of Hyderabad, Hyderabad, India. \\
		\vspace{.3in}
		
		March  06, 2024

	\end{center}

	\begin{abstract}
		In the following, we introduce new proportional hazard (PH) processes, which are derived by a marginal transformation applied to complementary power function distribution (CPFD) processes. Also, we introduce two new Pareto processes, which are derived from the proportional hazard family.  We discuss distributional features of such processes, explore inferential aspects and include an example of applications of the new processes to real-life data.
	\end{abstract}
	
	\noindent{\bf Key Words}:  proportional hazard process, power function distribution process,  complementary power function distribution process, Pareto process, stationarity, Markovian property
	
	\bigskip

	\bigskip
	
	\section{Introduction}
	An important tool that will be used in constructing proportional hazard processes is a distribution that we will call the complementary power function distribution (CPFD). The CPFD processes are the perfect tools for developing stationary proportional hazard processes, which will be of interest since an assumption of proportional hazards is frequently deemed to be plausible in many applications. Another important specific case of a proportional hazard family of distributions is given by the classical Pareto distribution which is used in the study of inequality in income distribution.  In the following sections we discuss distributional features of CPFD processes and the Pareto processes, explore inferential aspects and include an example of application of such processes to real-life data.

	\section{Proportional reversed hazard processes}
	It is well known that proportional reversed hazard (PRH) processes can be viewed as being obtained by marginal transformations applied to a particular PFD process. In the following, we define two PFD processes introduced in Kundu \cite{kd22} (called Kundu process) and Arnold et al. \cite{asm23} (called A-M process).
	\bigskip
	\subsection{PFD-processes} We begin by recalling the definition of the power function distribution (PFD) and two PFD processes introduced in Arnold et al. \cite{asm23}. We say that $X$ has a power function distribution (PFD) with parameter $\alpha$ if its distribution is of the form.
	
	\begin{equation}\label{PF-df}
		F_X(x)=x^{\alpha}, \ \ 0<x<1, \ \ \alpha>0.
	\end{equation}
	If $X$ has distribution ({\ref{PF-df}) then we will write $X\sim PFD(\alpha).$ The corresponding density is
		\begin{equation}
			f_X(x)=\alpha x^{\alpha-1}, \ \ 0<x<1, \ \ \alpha>0.
		\end{equation}
		This can be recognized as a $Beta(\alpha,1)$ density.
		There are two special properties of power function distributions that will be useful for defining stationary PFD processes. 
		
		\begin{description}		
			\item[(A) Closure under maximization:] If $X$ and $Y$ are independent with $X\sim PFD(\alpha)$ and $Y \sim PFD(\beta)$,  then $Z=\max\{X,Y\} \sim PFD(\alpha+\beta)$.
			\item[(B) Closure under raising to a power:]  If $X\sim PFD(\alpha)$ and if $\delta >0$, then $ V=X^{\delta} \sim PFD(\alpha/\delta)$.
		\end{description}
		\bigskip
		
		The following two stationary power function distribution processes (PFD-processes)  with PFD marginals proved to be useful for generating stationary proportional reversed hazard processes. 
		\begin{description}		\item[(I) A Kundu PFD process:] Consider an i.i.d. sequence, $\{U_n\}$, of  $uniform(0,1)$ random variables. The PFD process 
			$\left\{ X^*_n \right\} $ will be defined by
			\begin{equation}\label{KPFDP}
				X^*_n= \max \left\{U_n^{1/\alpha},U_{n-1}^{1/\beta} \right\},
			\end{equation}
			where $\alpha, \beta >0$, see Kundu \cite{kd22}.
			\item[(II) An A-M PFD process:]  Define $Y^*_0=U_0^{1/\alpha} \sim PF(\alpha)$ and for $n=1,2,...$ define the PFD process 
			
			\begin{equation}\label{AMPFDP}
				Y^*_n= \max \{{(Y^*_{n-1})}^{\alpha/(\alpha-\delta)},U_n^{1/\delta}\},
			\end{equation}
			where $\delta \in (0,\alpha)$ and the $U_n$'s are i.i.d. $uniform(0,1)$ random variables, see Arnold et al. \cite{asm23}.
		\end{description}
		
		\bigskip
		
		\bigskip
		
		Now, if $\{X_n\}$ is either one of the above PFD processes, we can construct a stationary proportional reverse hazard process of the form
		
		\begin{equation}
			Y_n=F_0^{-1}(X_n), \ \ \ n=0,1,2,....
		\end{equation}
		
		where $F_0$ is an arbitrary distribution function with support $(0,\infty)$.
		
		We refer to Arnold et al. \cite{asm23} for stochastic properties, inferential aspects, and also an example of applications of the above two new PFD processes. 
		
		\section{Proportional hazard processes}
		An important tool that will be used in constructing proportional hazard processes is a distribution that we will call the complementary power function distribution (CPFD). It might also be called a power function survival distribution. The construction of proportional hazard processes, using the CPFD will be a close parallel to that used in developing the proportional reverse hazard processes above.
		
		\subsection{The complementary power function distribution}
		
		We say that $X$ has a complementary power function distribution (CPFD) with parameter $\alpha$ if its distribution is of the form.
		
		\begin{equation}\label{CPF-df}
			F_X(x)=1-(1-x)^{\alpha}, \ \ 0<x<1, \ \ \alpha>0,
		\end{equation}
		i.e., if $\overline{F}_X(x)=(1-x)^{\alpha}, \ \ 0<x<1.$ \ \ 
		If $X$ has distribution (\ref{CPF-df}) then we will write $X\sim CPFD(\alpha).$ The corresponding density is
		\begin{equation}
			f_X(x)=\alpha (1- x)^{\alpha-1}, \ \ 0<x<1, \ \ \alpha>0.
		\end{equation}
		This can be recognized as a $Beta(1,\alpha)$ density.
		There are two special properties of complementary power function distributions that will be useful for defining stationary CPFD processes. 
		\begin{description}		
			\item[(C) Closure under minimization:] If $X$ and $Y$ are independent with $X\sim CPFD(\alpha)$ and $Y \sim CPFD(\beta)$,  then $Z=\min\{X,Y\} \sim CPFD(\alpha+\beta)$.
			\item[(D) Closure under complementary raising to a power:]  If $X\sim CPFD(\alpha)$ and if $\delta >0$, then $ V=1-(1-X)^{\delta} \sim CPFD(\alpha/\delta)$. Of course, this is just a rewritten version of statement (B) for PFD variables.
		\end{description}
		\bigskip
		\subsection {Two CPFD-processes}

		Note that, if $U\sim Uniform(0,1)$ and $X=1-U^{1/\alpha}$ then $P(X>x)=(1-x)^\alpha$.
		
		\bigskip
		
		Now, analogous to the two processes defined in Section 2.1, we have the following CPFD processes.

		\begin{description}
			\item[(I) A Kundu CPFD process:] Consider an i.i.d. sequence, $\{U_n\}$, of  $uniform(0,1)$ random variables. The Kundu CPFD-process 
			$\left\{ V_n^{*} \right\} $ will be defined by
			\begin{equation}\label{CPFD-proc}
				V_n^{*} = \min\{1-U_{n-1}^{1/\alpha},1-U_n^{1/\beta}\}
			\end{equation}
			where $\alpha, \beta >0$.
			
			Note that $V_n^{*}=1-X^{*}_n, \forall n$ where $\{X_n^{*}\}$ is the Kundu PFD process .
			
			\item[(II) An A-M CPFD process:]  We define the A-M CPFD process $\{W_n^{*}\}$ as follows.
			Define $W^*_0=1-U_0^{1/\alpha}$ and for $n=1,2,...$ define the CPFD process 
			
			\begin{equation}\label{AMPFDP}
				W^*_n= \min \{1-{(1-W^*_{n-1})}^{\alpha/(\alpha-\delta)},1-U_n^{1/\delta}\},
			\end{equation}
			where $\delta \in (0,\alpha)$ and the $U_n$'s are i.i.d. $uniform(0,1)$ random variables. 
		\end{description}
		Note that $W_n^{*}=1-Y^{*}_n \  \forall n$ where $\{Y_n^{*}\}$ is the A-M PFD process .
		\bigskip
		\
		
		\subsection{Two proportional hazard processes}
		The motivation for introducing complementary power function distribution processes (which are so simply related to power function distribution processes) is that the CPFD processes are the perfect tools for developing stationary proportional hazard processes, which will be of interest since an assumption of proportional hazards is frequently deemed to be plausible in many applications. A simple marginal transformation applied to a CPFD process will be all that is required.

		We can then define a Kundu PH process $\{S_n^{*}\}$ with PH marginals using a Kundu CPFD process $\{V_n^{*}\}$ as follows
		\begin{equation}
			S_n^{*}=F_0^{-1}(V_n^{*}), \ \ \ n=0,1,2,....
		\end{equation}
		where it can be verified that $$P(S_n^{*} >s)=[\overline{F}_0(s)]^{\alpha + \beta}.$$
		
		\bigskip

		In parallel fashion, we can define an  A-M PH process as follows
		\begin{eqnarray}
			T_n^{*} = F^{-1}_0(W_n^{*}) = F^{-1}_0(1-Y^*_n) 
		\end{eqnarray}
		where ${W_n}$ is the A-M CPFD process.
		
		We may readily verify that $\{T_n^{*}\}$ has proportional hazard marginal distributions as follows:
		\begin{eqnarray*}
			P(T_n^{*}>t)&=&P(F^{-1}_0(1-Y^*_n) >t)=P(1-Y_n^{*}>F_0(t) \\
			&=&P(Y_n^{*}<1-F_0(t))=P(Y_n^{*}<\overline{F_0}(t)) \\
			&=&[\overline{F_0}(t)]^{\alpha}.
		\end{eqnarray*}
		
		\bigskip

		\subsection{New stationary Pareto processes}
		A specific case of a proportional hazard family of distributions is provided by the classical Pareto distribution, which is much used in the study of inequality in economic variables such as income \cite{bca15}. Recall that a random variable $X$ has a classical Pareto distribution or Pareto(I)  distribution if its survival function is of the form
		$$\overline{F}_X(x)=(x/\sigma)^{-\alpha}, \ \ x>\sigma.$$
		
		\bigskip
		
		To indicate this, we write $X\sim P(I)(\sigma,\alpha)$.
		
		\bigskip
		
		Note that if $X_1 \sim P(I)(\sigma,\alpha_1)$ and $X_2 \sim P(I)(\sigma,\alpha_2)$ are independent 
		and if we define $Z=\min \{X_1,X_2\}$ then $Z \sim P(I)(\sigma, \alpha_1+\alpha_2)$.
		
		\bigskip
		
		Note that the family of $P(I)(\sigma,\alpha)$ distributions, for a fixed value of $\sigma$ is a
		proportional-hazard family with its basic survival function of the form $\overline{F_0}(x)=(x/\sigma)^{-1},$ with corresponding quantile function
		$$F_0^{-1}(y)=\sigma (1-y)^{-1}.$$
		
		\bigskip
		
		It will thus be possible to define two stationary Pareto processes as in Subsection 3.3 using this particular form for $\overline{F_0}(x)$. Instead, we will equivalently develop the processes by utilizing a simple relationship between Pareto and uniform distributions.
		
		Recall that the basic proportional hazard distribution (or CPFD) family has a density  of the form:
		$$f_0(x:\alpha)=\alpha(1-x)^{\alpha-1}, \ \ 0<x<1,$$
		
		\bigskip
		
		which is the density of a random variable of the form $Y=1-U^{1/\alpha}$
		
		\bigskip
		
		where $U \sim uniform(0,1).$ 
		\bigskip
		
		Note that a $Pareto(I)(\sigma,\alpha)$ variable $X$ can be represented in the form
		
		$$X=\sigma U^{-1/\alpha},$$
		or equivalently
		$$X=\sigma(1- U)^{-1/\alpha}.$$
		\bigskip
		
		Now, we may define two new Pareto processes as follows:
		
		\bigskip
		
		\begin{description}
			\item[(I) A Kundu Pareto process:] $S^{*}_n=\sigma /X^*_n$ where $\{X^*_n\}$ is a Kundu PFD process. i.e.,
			\begin{eqnarray}
				S_n^{*}=\sigma \min\{U_{n-1}^{-1/\alpha},U_n^{-1/\beta}\},
			\end{eqnarray}
			where $\alpha, \beta >0$.
			For this process we have $S_n^{*}\sim P(\sigma,\alpha+\beta) \ \forall n.$ 
			We refer to Figures \ref{fig1} for examples of simulated sample paths of steps  $n=10$ and $n=50$.
			\begin{figure} [h!]
				\begin{subfigure}{8cm}
					\centering
					\includegraphics[width=8cm]{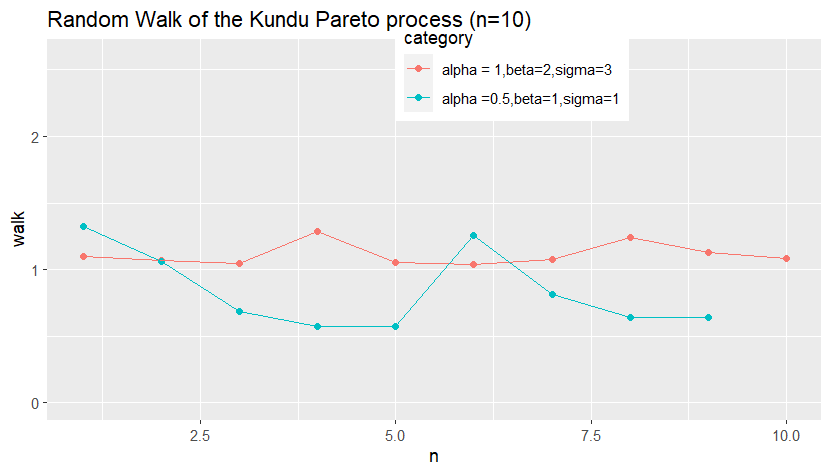}
					\caption{$n=10$} 
				\end{subfigure}
				\begin{subfigure}{8cm}
					\centering
					\includegraphics[width=8cm]{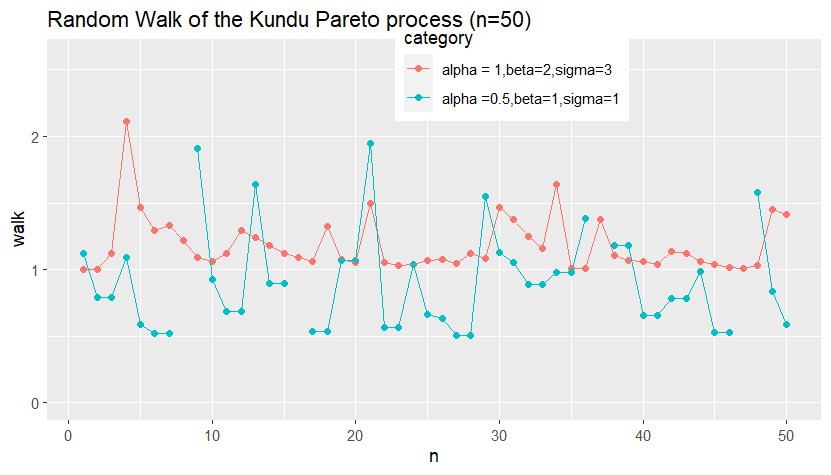}
					\caption{$n=50$} 
				\end{subfigure}
				\caption{Sample Path for the Kundu Pareto process }
				\label{fig1}
			\end{figure}
			\item[(II) An A-M Pareto process:]  define $$T^{*}_{-1}\sim Pareto(I)(\sigma,\alpha)$$
			and
			\begin{eqnarray}
				T_n^{*} = \sigma \min\{{(T^*_{n-1}/\sigma)}^{-\alpha/(\alpha-\delta)}, U_n^{-1/\delta}\}.
			\end{eqnarray}
			where $0<\delta<\alpha.$  For this process we have $T_n^{*} \sim P(\sigma,\alpha).$
		\end{description}
		
		We refer to Figures \ref{fig2} for examples of simulated sample paths of steps  $n=10$ and $n=50$.
		\begin{figure} [h!]
			\begin{subfigure}{8cm}
				\centering
				\includegraphics[width=8cm]{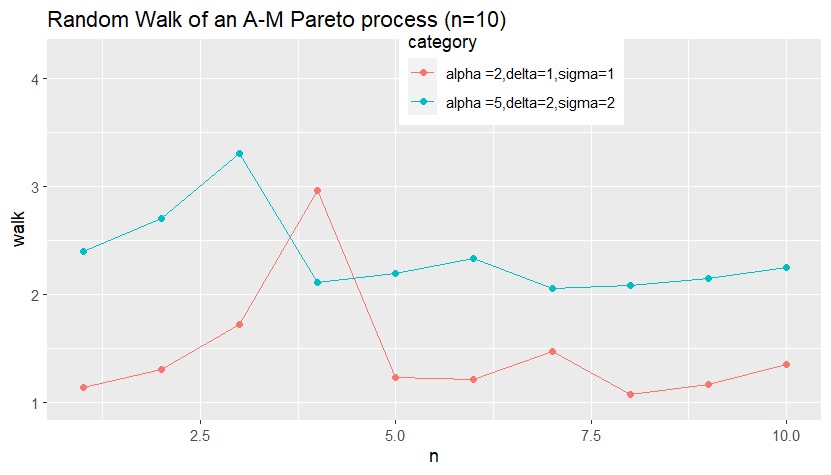}
				\caption{$n=10$} 
			\end{subfigure}
			\begin{subfigure}{8cm}
				\centering
				\includegraphics[width=8cm]{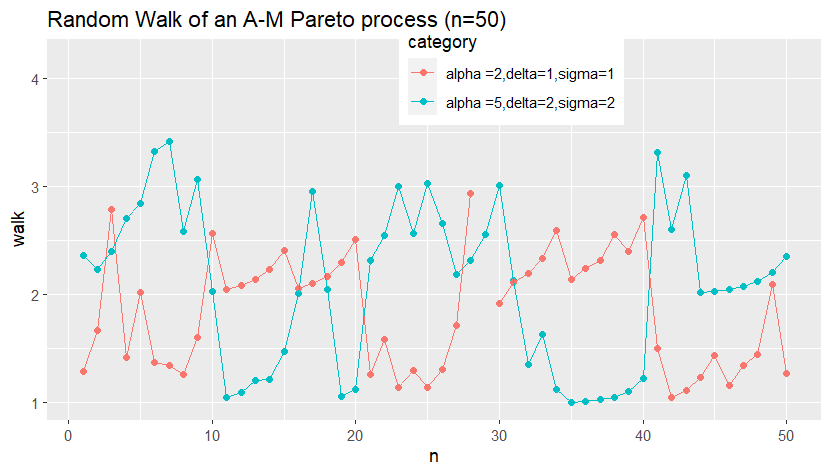}
				\caption{$n=50$} 
			\end{subfigure}
			\caption{Sample Path for an A-M Pareto process }
			\label{fig2}
		\end{figure}
		
		\section{Stochastic properties of the CPFD and Pareto processes}
		\subsection{Kundu CPFD process}
		\begin{theorem}
			If the sequence of random variables $\{V^*_n\}$ is as defined in equation $(3.3)$, then the following statements are true:
			\begin{enumerate}[label=(\alph*)]
				\item For each n, $\{V^*_n\}$ has a CPFD$(\alpha+\beta)$ distribution. 
				\item The joint distribution function of $(V^*_{n-1}, V_n^*)$ is $1- (1-v_{n-1})^{\alpha +\beta} -(1-v_n)^{\alpha + \beta} + 
				(1-v_{n-1})^{\alpha} (1-v_n)^{\beta} \min\{(1-v_n)^\alpha, (1-v_{n-1})^\beta\}$.
				\item $\{V^*_n\}$  is a stationary non-Markovian process.
			\end{enumerate}
		\end{theorem}
		
		\begin{proof}
			If $V^*_n$ is a CPFD process, then $V^*_n=1-X^*_n$, where $X^*_n$ is a stationary non-Markovian PFD process (see Arnold et al. \cite{asm23}). Hence, the stationarity and non-Markovian nature of $V^*_n$ follows. 
			To prove part $(a)$ we use the following notation
			\begin{eqnarray}
				F_{V^*_n}(v_n) = P(V_n^* \leq v) &=& P(1- X^*_n \leq v_n) \nonumber \\
				&=& P(X^*_n \geq 1-v_n) \nonumber \\
				&=& 1- P(X^*_n \leq 1- v_n) \nonumber \\
				&=& 1-(1-v_n)^{\alpha +\beta}
			\end{eqnarray}
			For the proof of Part $(b)$, we refer to $V^*n=1-X^*_n$ and the analogous result for $X^*_n$ in Arnold et al. \cite{asm23}. Now, the joint distribution of $V^*_{n-1}$ and $V^*_n$,
			\begin{eqnarray}
				F_{V^*_n, V^*_{n-1}} (v_n,v_{n-1} ) &=& P(V^*_n \leq v_n, V^*_{n-1} \leq v_{n-1}) \nonumber\\
				&=& P(1-X^*_n \leq v_n, 1-X^*_{n-1} \leq v_{n-1})  \nonumber\\
				&=& P(X^*_n \geq 1-v_n, X^*_{n-1} \geq 1- v_{n-1}) \nonumber \\
				&=& 1- F_{X^*_{n-1}}(1-v_{n-1}) -F_{X^*_{n}}(1-v_n) + F_{X^*_n,X^*_{n-1}}(1-v_n,1-v_{n-1})
				\nonumber\\
				&=& 1- (1-v_{n-1})^{\alpha + \beta} -(1-v_n)^{\alpha + \beta} + (1-v_n)^\alpha  \nonumber\\ &&(1-v_{n-1})^\beta \min\{ (1-v_n)^\alpha, (1-v_{n-1})^\beta\}.
			\end{eqnarray}
		\end{proof}
		
		\begin{corollary}
			If the sequence of random variables $\{V^*_n\}$ is as defined in equation $(3.3)$, then the mean and variance are
			\begin{eqnarray}
				E(V^*_n) &=& \frac{1}{\alpha +\beta +1} \\
				Var(V^*_n) &=& \frac{\alpha +\beta}{(\alpha + \beta +1)^2(\alpha +\beta +2)}	\end{eqnarray}
			
			\begin{proof}
				We know that 
				\begin{eqnarray}
					E(V^*_n) &=& E( 1- X^*_n) = 1- E(X^*_n) = 1- \frac{\alpha +\beta}{\alpha + \beta +1} = \frac{1}{\alpha +\beta +1} \\
					Var(V^*_n) &=&  Var(1-X^*_n) = Var(X_n^*) = \frac{\alpha +\beta}{(\alpha +\beta +1)^2 (\alpha +\beta +2)}.
				\end{eqnarray}
			\end{proof}
		\end{corollary}

		If the sequence of random variables $\{V^*_n\}$ is as defined in equation $(3.3)$, the auto-correlation function is
		\begin{eqnarray}
			Corr(V^*_n, V^*_{n-1}) &=& \frac{Cov(V^*_n, V^*_{n-1})}{\sqrt{Var(V^*_n) Var(X^*_{n-1})}} \nonumber \\
			&=&  \frac{E(V^*_n V^*_{n-1})- \big(\frac{1}{\alpha +\beta +1}\big)^2}{\frac{\alpha + \beta}{(\alpha +\beta +1)^2 (\alpha +\beta +2)}}
		\end{eqnarray}
		
		From the usual expression for the auto-correlation in $(4.7)$, it is evident that the detail lacking is an expression for $E(V^*_n V^*_{n-1})$.
		To evaluate this, we argue as follows
		\begin{eqnarray}
			E(V^*_n V^*_{n-1}) &=& E((1-X^*_n)(1-X^*_{n-1}))  \nonumber \\
			&=& E(1-X^*_{n-1} -X^*_n + X^*_n X^*_{n-1})  \nonumber \\
			&=& 1 - 2 \frac{1}{\alpha +\beta +1} + E(X^*_n X^*_{n-1}).
		\end{eqnarray}
		For the computation of $ E(X^*_n X^*_{n-1})$ we refer to page 8 of Arnold et al. \cite{asm23}, hence
		\begin{eqnarray}
			E(X^*_n X^*_{n-1}) = A + B + C + D
		\end{eqnarray}
		where
		\begin{eqnarray}
			A &=& \frac{\beta^2}{\beta +1} \Bigg[ \frac{1}{\alpha + \beta +1} - \frac{\alpha}{\beta^2 + \alpha^2 +  \alpha \beta + \alpha +\beta}\Bigg], \nonumber \\
			B &=& \left[\frac{\alpha+\beta+\alpha^2+\beta^2+\alpha\beta}{\alpha\beta}\right]^{-1},\nonumber\\
			C &=& \frac{\beta \alpha}{(\beta +1)(\alpha +1)} \Bigg[ 1 - \frac{ \beta}{\alpha + \beta + 1 } - \frac{ \alpha}{\alpha + \beta + 1 }
			+ \frac{\alpha \beta}{\alpha^2 + \beta^2 +  \alpha \beta + \alpha + \beta} \Bigg], \nonumber\\
			D &=& \frac{\alpha^2 \beta}{\alpha +1} \Bigg[ \frac{1}{\beta^2 + \alpha + \alpha \beta} - \frac{1}{\alpha^2 + \beta^2  + 2 \alpha + \alpha \beta}\Bigg]. \nonumber
		\end{eqnarray}
		
		\begin{theorem}
			If the sequence of random variable $\{V^*_n\}$ is as defined in equation $(3.3)$ then 
			\begin{eqnarray}
				P(V^*_1<V^*_2) = \left\{
				\begin{array}{ll}
					\frac{\alpha + \beta}{2\alpha+\beta} & \mbox{if } \alpha > \beta \\
					\frac{\beta}{2\beta+\alpha} & \mbox{if } \alpha \leq \beta. 
				\end{array}
				\right.
			\end{eqnarray}
		\end{theorem}
		\begin{proof}
			By definition of the process, we have $V^*_1 = 1-X^*_1$ and $V^*_2 = 1-X^*_2$. 
			\begin{eqnarray}
				P(V^*_1 < V^*_2) = P(1-X^*_1 < 1-X^*_2) =P(X^*_1 > X^*_2)
			\end{eqnarray}
			So, if $\alpha > \beta$ we have
			\begin{eqnarray}
				P(V^*_1 < V^*_2) =\frac{\alpha + \beta }{2 \alpha + \beta}
			\end{eqnarray}
			Now, for $\alpha < \beta$
			\begin{eqnarray}
				P(V^*_1 > V^*_2) =\frac{\alpha + \beta}{2 \beta + \alpha}
			\end{eqnarray}
		\end{proof}
		
		\subsection{A-M CPFD process}
		
		\begin{theorem}
			If the sequence of random variables $\{W^*_n\}$ is as defined in equation $(3.4)$, then the following statements are true:
			\begin{enumerate}[label=(\alph*)]
				\item For each n, $\{W^*_n\}$ has a CPFD$(\alpha)$ distribution. 
				\item The joint distribution function of $(W^*_{n-1}, W_n^*)$ is $1- (1-w_{n-1})^{\alpha } -(1-w_n)^{\alpha } + 
				(1-w_{n})^{\delta}  \min\{(1-v_n)^{\alpha -\delta}, (1-v_{n-1})^\alpha\}$.
				\item $\{W^*_n\}$  is a stationary Markovian process.
			\end{enumerate}
		\end{theorem}
		
		\begin{proof}
			We use the property that , if $\{W_n^*\}$ is an A-M CPFD process  
			$W^*_n = 1- Y^*_n$, where $\{Y^*_n\}$ is an A-M PFD process.	To prove part $(a)$ we use the following notation
			\begin{eqnarray}
				F_{W^*_n}(w) &=& P(W^*_n \leq w) = P(1-Y^*_n  \leq w) = P(Y^*_n \geq 1-w) = 1- P(Y^*_n \leq 1-w) \nonumber \\
				&=& 1-(1-w)^\delta P(Y^*_{n-1} \leq (1-w)^{\frac{\alpha -\delta}{\alpha}})
			\end{eqnarray}
			Hence, using mathematical induction, we have
			\begin{eqnarray}
				F_{W^*_n}(w) = w^\alpha.
			\end{eqnarray}
			Now, for Part (b) of the theorem, the joint distribution of $W^*_n$ and $W^*_{n+1}$ is, by stationarity of the process, the same as the joint distribution of $W^*_0$ and $W^*_1$ which may be computed as follows:
			
			\begin{eqnarray}
				F_{W^*_0,W^*_1} (w_0,w_1) &=& P(W^*_0 \leq w_0, W^*_1 \leq w_1) \nonumber\\
				&=& P(1-Y^*_0 \leq w_0, 1-Y^*_1 \leq w_1)  \nonumber \\
				&=& P(Y^*_0 \geq 1-w_0, Y^*_1 \geq 1-w_1)  \nonumber \\
				&=& 1 - F_{Y^*_0}(1-w_0) - F_{Y^*_1}(1-w_1) + F_{Y^*_0,Y^*_1}(1-w_0,1-w_1) \nonumber \\
				&=& 1 - (1-w_0)^\alpha - (1-w_1)^\alpha + (1-w_1)^\delta \min\{ w^\alpha_0, w^{\alpha -\delta}_1\}.
			\end{eqnarray}
			It remains to prove Part (c). Stationarity is trivial from the definition of the process. For the Markovian property, recall that $W_n^{*}=1-Y^{*}_n$, we know that $Y^*_n$ is Markovian, and $W^*_n$ is invertible transformation of a Markovian process, hence $W^*_n$ is also  Markovian.
		\end{proof}
		\begin{corollary}
			If the sequence of random variable $\{W^*_n\}$ is as defined in equation $(3.4)$, then the mean and variance are 
			\begin{eqnarray}
				E(W^*_n) &=&   \frac{1}{\alpha +1} \\
				Var(W^*_n) &=& \frac{\alpha}{(\alpha +1)^2(\alpha +2)}.
			\end{eqnarray}
		\end{corollary}
		\begin{proof}
			We  known that 
			\begin{eqnarray}
				E(W^*_n) &=& E(1-Y^*_n) = 1 - E(Y^*_n) = 1 - \frac{\alpha}{\alpha +1} = \frac{1}{\alpha +1} \\
				Var(W^*_n) &=& Var(1-Y^*_n) = Var(Y^*_n) = \frac{\alpha}{(\alpha +1)^2(\alpha +2)}.
			\end{eqnarray}
		\end{proof}
		
		If the sequence of random variables $\{W^*_n\}$ is as defined in equation $(3.4)$ then the auto-correlation function will be 
		\begin{eqnarray}
			Corr(W^*_0, W^*_1) = \frac{E(W^*_0 W^*_1) - \big(\frac{1}{\alpha+1}\big)^2}{\frac{\alpha}{(\alpha +1)^2(\alpha +2)}}.
		\end{eqnarray}
		which is monotone in $\delta$.

		The only thing needed in order to get the auto-correlation is $E(W^*_0 W^*_1)$. Now,
		\begin{eqnarray}
			E(W^*_0 W^*_1) &=& E((1-Y^*_0)(1-Y^*_1) ) = 1- E(Y^*_0) - E(Y^*_1) + E(Y^*_0 Y^*_1) \nonumber \\
			&=& 1- 2 \frac{1}{\alpha +1} + \frac{\delta}{\delta +1} \Bigg[ \frac{\alpha}{\alpha +1} - \Bigg(\frac{1}{\alpha}+\frac{\delta+1}{\alpha-\delta} +1 \Bigg)^{-1}\Bigg] + \Bigg[\frac{1}{\alpha} + \frac{1+ \delta}{\alpha -\delta} +1\Bigg]^{-1}. \nonumber \\
		\end{eqnarray}
		For the computation of $E(Y^*_0 Y^*_1)$ we refer to page 18 of Arnold et al. \cite{asm23}.
		
		\begin{theorem}
			If the sequence of random variables $\{W^*_n\}$ is as defined in equation $(3.4)$ then 
			\begin{eqnarray}
				P(W^*_1 < W^*_1) = \frac{\delta}{\alpha + \delta}.
			\end{eqnarray}	
		\end{theorem}
		\begin{proof}
			By the definition of the process we have $W_1^*= 1-Y_1^*$, consequently $P(W^*_1< W^*_1)$ is given by
			\begin{eqnarray}
				P(W^*_1 < W^*_0) = P(1-Y^*_1< 1-Y^*_0)= P(Y^*_1 > Y^*_0) = \frac{\delta}{\alpha +\delta}
			\end{eqnarray}
			For the computation of $P(Y^*_1  >Y^*_0)$ we refer to page 19 of Arnold et al. \cite{asm23}.
		\end{proof}

		%%%% Pareto process
		
		\subsection{Kundu Pareto process}
		\begin{theorem}
			If the sequence of random variables $\{S^*_n\}$ is as defined in equation $(3.9)$, then the following statements are true:
			\begin{enumerate}[label=(\alph*)]
				\item For each n, $\{S^*_n\}$ has a Pareto$(\alpha+\beta)$ distribution. 
				\item The joint distribution function of $(S^*_{n-1}, S_n^*)$ is 
				\begin{eqnarray}1- \Big (\frac{s_{n-1}}{\sigma}\Big )^{-(\alpha +\beta)} -\Big (\frac{s_n}{\sigma}\Big )^{-(\alpha +\beta)} +     \Big (\frac{s_n}{\sigma}\Big )^{-\alpha} \Big (\frac{s_{n-1}}{\sigma}\Big )^{-\beta}\min \Bigg\{  \Big (\frac{s_n}{\sigma}\Big )^{-\alpha}, \Big (\frac{s_{n-1}}{\sigma}\Big )^{-\beta}\Bigg\}. \nonumber
				\end{eqnarray}
				\item $\{S^*_n\}$  is a stationary non-Markovian process.
			\end{enumerate}
		\end{theorem}
		
		\begin{proof}
			If $S^*_n$ is a Pareto process, then $S^*_n=\sigma/X^*_n$, where $X^*_n$ is a stationary non-Markovian PFD process (see Arnold et al. \cite{asm23}). Hence, the stationarity and non-Markovian nature of $S^*_n$ follows. 
			To prove part $(a)$ we use the following notation
			\begin{eqnarray}
				F_{S^*_n}(s_n) = P(S_n^* \leq s_n) &=& P(\sigma/ X^*_n \leq s_n) \nonumber \\
				&=& P(X^*_n \geq \sigma/s_n) \nonumber \\
				&=& 1- P(X^*_n \leq \sigma/ s_n) \nonumber \\
				&=& 1-\Big (\frac{s_n}{\sigma}\Big )^{-(\alpha +\beta)}
			\end{eqnarray}
			For the proof of Part $(b)$, we refer to $S^*n=\sigma/X^*_n$ and the analogous result for $S^*_n$ in Arnold et al. \cite{asm23}. Now, the joint distribution of $S^*_{n-1}$ and $S^*_n$,
			\begin{eqnarray}
				F_{S^*_n, S^*_{n-1}} (s_n,s_{n-1} ) &=& P(S^*_n \leq s_n, S^*_{n-1} \leq s_{n-1}) \nonumber\\
				&=& P(\sigma/X^*_n \leq s_n, \sigma/X^*_{n-1} \leq s_{n-1})  \nonumber\\
				&=& P(X^*_n \geq \sigma/s_n, X^*_{n-1} \geq  \sigma/s_{n-1}) \nonumber \\
				&=& 1- F_{X^*_{n-1}}(\sigma/s_{n-1}) -F_{X^*_{n}}(\sigma /s_n) + F_{X^*_n,X^*_{n-1}}(\sigma/s_n,\sigma/s_{n-1})
				\nonumber\\
				&=& 1- \Big (\frac{s_{n-1}}{\sigma}\Big )^{-(\alpha +\beta)} -\Big (\frac{s_n}{\sigma}\Big )^{-(\alpha +\beta)} +   \nonumber\\ &&   \Big (\frac{s_n}{\sigma}\Big )^{-\alpha} \Big (\frac{s_{n-1}}{\sigma}\Big )^{-\beta}\min \Bigg\{  \Big (\frac{s_n}{\sigma}\Big )^{-\alpha}, \Big (\frac{s_{n-1}}{\sigma}\Big )^{-\beta}\Bigg\}.
			\end{eqnarray}
		\end{proof}
		
		\begin{corollary}
			If the sequence of random variables $\{S^*_n\}$ is as defined in equation $(3.9)$, then the mean (for $\alpha + \beta >1$)and variance (for $\alpha  +\beta >2$) are
			\begin{eqnarray}
				E(S^*_n) &=& \frac{\sigma(\alpha+\beta)}{\alpha +\beta -1} \\
				Var(S^*_n) &=& \frac{\sigma^2(\alpha +\beta)}{(\alpha + \beta -1)^2(\alpha +\beta -2)}	\end{eqnarray}
			
			\begin{proof}
				We know that 
				\begin{eqnarray}
					E(S^*_n) &=& E( \sigma/ X^*_n) = \sigma E(1/X^*_n) = \sigma \frac{\alpha+\beta}{\alpha + \beta -1} \\
					Var(S^*_n) &=&  Var(\sigma/X^*_n) =  \sigma^2 Var(1/X_n^*) = \sigma^2 \frac{\alpha +\beta}{(\alpha +\beta -1)^2 (\alpha +\beta -2)}.
				\end{eqnarray}
			\end{proof}
		\end{corollary}

		If the sequence of random variables $\{S^*_n\}$ is as defined in equation $(3.9)$, the auto-correlation function is
		\begin{eqnarray}
			Corr(S^*_n, S^*_{n-1}) &=& \frac{Cov(S^*_n, S^*_{n-1})}{\sqrt{Var(S^*_n) Var(S^*_{n-1})}} \nonumber \\
			&=&  \frac{E(S^*_n S^*_{n-1})- \big(\sigma \frac{\alpha+\beta}{\alpha + \beta -1}\big)^2}{\sigma^2 \frac{\alpha +\beta}{(\alpha +\beta -1)^2 (\alpha +\beta -2)}}
		\end{eqnarray}
		
		From the usual expression for the auto-correlation in $(4.7)$, it is evident that the detail lacking is an expression for $E(S^*_n S^*_{n-1})$.
		To evaluate this, we argue as follows
		\begin{eqnarray}
			E(S^*_n S^*_{n-1}) &=& E((\sigma/X^*_n)(\sigma/X^*_{n-1}))  \nonumber \\
			&=& \sigma^2 E(1/X^*_{n-1} 1/X^*_{n})  \nonumber \\.
		\end{eqnarray}
		For the computation of $E(1/X^*_{n-1} 1/X^*_{n})$ 
		%==================Edited========22feb======1pm======
		\begin{eqnarray}
			E(1/X^*_{n-1} 1/X^*_{n}) &=&E(V^{-1/\beta}W^{-1/\beta} I(U<V^{\alpha/\beta},W>V^{\beta/\alpha} ) ) \nonumber\\ 
			& &+E(V^{-1/\beta}V^{-1/\alpha} I(U<V^{\alpha/\beta},W<V^{\beta/\alpha})) \nonumber \\
			& &+E(U^{-1/\alpha}W^{-1/\beta} I(U>V^{\alpha/\beta},W>V^{\beta/\alpha}) )\nonumber \\
			& &+E(U^{-1/\alpha}V^{-1/\beta} I(U>V^{\alpha/\beta},W<V^{\beta/\alpha}) ) \nonumber\\
			& & = A+B+C+D
		\end{eqnarray}
		where
		\begin{eqnarray*}
			X^*_n &=&\max\{U^{1/\alpha},V^{1/\beta} \}\\
			X^*_{n-1} &=&\max\{V^{1/\alpha},W^{1/\beta} \}
		\end{eqnarray*}
		and expressions for A, B, C, and D are as follows.
		\begin{eqnarray}
			A &=&\int_0^1  dv \int_{0}^{v^{\alpha/\beta}}du \int_{v^{\beta/\alpha}}^{1} dw \ [v^{-1/\beta} w^{-1/\beta} ]  \nonumber\\
			&=& \frac{\beta^2}{\beta-1} \Bigg[ \frac{1}{\alpha + \beta-1} - \frac{\alpha}{\alpha^2+\beta^2  +  \alpha \beta -\alpha- \beta}\Bigg], \nonumber \
		\end{eqnarray}
		
		\begin{eqnarray*}
			B &=&\int_0^1  dv \int_{0}^{v^{\alpha/\beta}}du \int_{0}^{v^{\beta/\alpha}} dw \ [v^{-1/\alpha} v^{-1/\beta} ]  \nonumber\\
			&=& \left[\frac{\alpha^2+\beta^2+\alpha\beta - \alpha-\beta}{\alpha\beta}\right]^{-1},\nonumber\\
		\end{eqnarray*}
		
		\begin{eqnarray*}
			C &=&\int_0^1  dv \int_{v^{\alpha/\beta}}^{1}du \int_{v^{\beta/\alpha}}^{1} dw \ [u^{-1/\alpha} w^{-1/\beta} ]  \nonumber\\
			&=& \frac{\alpha\beta}{(\beta - 1)(\alpha-1)} \Bigg[ 1 - \frac{ \beta}{\alpha + \beta - 1 } - \frac{ \alpha}{\alpha + \beta - 1 }
			+ \frac{\alpha \beta}{1 +  2\alpha \beta - \alpha - \beta} \Bigg], \nonumber\\
		\end{eqnarray*}
		
		\begin{eqnarray*}
			D &=&\int_0^1dv \int_{v^{\alpha/\beta}}^{1}du \int_{0}^{v^{\beta/\alpha}}dw \ [v^{-1/\alpha} u^{-1/\beta} ]  \nonumber\\
			&=& \frac{\alpha^2 }{\alpha - 1} \Bigg[ \frac{1}{ \alpha + \beta -1} - \frac{\beta}{\alpha^2 + \beta^2 - \alpha-\beta}\Bigg]. \nonumber
		\end{eqnarray*}

		%=================== :Edited upto here ========22feb======3pm IST======
		
		\begin{theorem}
			If the sequence of random variable $\{S^*_n\}$ is as defined in equation $(3.9)$ then 
			\begin{eqnarray}
				P(S^*_1<S^*_2) = \left\{
				\begin{array}{ll}
					\frac{\alpha + \beta}{2\alpha+\beta} & \mbox{if } \alpha > \beta \\
					\frac{\beta}{2\beta+\alpha} & \mbox{if } \alpha \leq \beta. 
				\end{array}
				\right.
			\end{eqnarray}
		\end{theorem}
		\begin{proof}
			By definition of the process, we have $S^*_1 = \sigma/X^*_1$ and $S^*_2 = \sigma/X^*_2$. 
			\begin{eqnarray}
				P(S^*_1 < S^*_2) = P(\sigma/X^*_1 < \sigma/X^*_2) =P(X^*_1 > X^*_2)
			\end{eqnarray}
			So, if $\alpha > \beta$ we have
			\begin{eqnarray}
				P(S^*_1 < S^*_2) =\frac{\alpha + \beta }{2 \alpha + \beta}
			\end{eqnarray}
			Now, for $\alpha < \beta$
			\begin{eqnarray}
				P(S^*_1 > S^*_2) =\frac{\alpha + \beta}{2 \beta + \alpha}
			\end{eqnarray}
		\end{proof}

		%%% A-M Pareto process
		
		\subsection{A-M  Pareto process}
		\begin{theorem}
			If the sequence of random variables $\{T^*_n\}$ is as defined in equation $(3.10)$, then the following statements are true:
			\begin{enumerate}[label=(\alph*)]
				\item For each n, $\{T^*_n\}$ has a Pareto$(\sigma,\alpha)$ distribution. 
				\item The joint distribution function of $(T^*_{n-1}, T_n^*)$ is 
				\begin{eqnarray}
					1 - (w_0/\sigma)^{-\alpha} - (w_1/\sigma)^{-\alpha} + (w_1/\sigma)^{-\delta} \min\{ (w_0/\sigma)^{-\alpha} , (w_1/\sigma)^{-(\alpha /\delta)}\} \nonumber.
				\end{eqnarray}
				\item $\{T^*_n\}$  is a stationary Markovian process.
			\end{enumerate}
		\end{theorem}

		\begin{proof}
			We use the property that , if $\{T_n^*\}$ is an A-M Pareto process  
			$T^*_n = \sigma/Y^*_n$, where $\{Y^*_n\}$ is an A-M PFD process.	To prove part $(a)$ we use the following notation
			\begin{eqnarray}
				F_{T^*_n}(t) &=& P(T^*_n \leq t) = P(\sigma/Y^*_n  \leq t) = P(Y^*_n \geq \sigma/t) = 1- P(Y^*_n \leq \sigma/t) \nonumber \\
				&=& 1-(\sigma/t)^\delta P\Big(Y^*_{n-1} \leq (\sigma/t)^{\frac{\alpha -\delta}{\alpha}}\Big)
			\end{eqnarray}
			Hence, using mathematical induction, we have
			\begin{eqnarray}
				F_{T^*_n}(t) = \Big(\frac{t}{\sigma} \Big)^{-\alpha}.
			\end{eqnarray}
			Now, for Part (b) of the theorem, the joint distribution of $T^*_n$ and $T^*_{n+1}$ is, by stationarity of the process, the same as the joint distribution of $T^*_0$ and $T^*_1$ which may be computed as follows:
			
			\begin{eqnarray}
				F_{T^*_0,T^*_1} (t_0,t_1) &=& P(T^*_0 \leq t_0, T^*_1 \leq t_1) \nonumber\\
				&=& P(\sigma/Y^*_0 \leq w_0, \sigma/Y^*_1 \leq w_1)  \nonumber \\
				&=& P(Y^*_0 \geq \sigma/w_0, Y^*_1 \geq \sigma/w_1)  \nonumber \\
				&=& 1 - F_{Y^*_0}(\sigma/w_0) - F_{Y^*_1}(\sigma/w_1) + F_{Y^*_0,Y^*_1}(\sigma/w_0,\sigma/w_1) \nonumber \\
				&=& 1 - (w_0/\sigma)^{-\alpha} - (w_1/\sigma)^{-\alpha} + (w_1/\sigma)^{-\delta} \min\{ (w_0/\sigma)^{-\alpha} , (w_1/\sigma)^{-(\alpha -\delta)}\} \nonumber \\.
			\end{eqnarray}
			It remains to prove Part (c). Stationarity is trivial from the definition of the process. For the Markovian property, recall that $T_n^{*}=\sigma/Y^{*}_n$, we know that $Y^*_n$ is Markovian, and $T^*_n$ is the invertible transformation of a Markovian process, hence $T^*_n$ is also  Markovian.
		\end{proof}
		\begin{corollary}
			If the sequence of random variable $\{T^*_n\}$ is as defined in equation $(3.10)$, then the mean (for $\alpha>1$) and variance (for $\alpha>2$) are 
			\begin{eqnarray}
				E(T^*_n) &=&   \frac{\sigma \alpha}{\alpha -1} \\
				Var(T^*_n) &=& \frac{\sigma^2\alpha}{(\alpha -1)^2(\alpha -2)}.
			\end{eqnarray}
		\end{corollary}
		\begin{proof}
			We  know that 
			\begin{eqnarray}
				E(T^*_n) &=& E(\sigma/Y^*_n) = \sigma E(1/Y^*_n)  = \sigma \frac{1}{\alpha -1} \\
				Var(T^*_n) &=& Var(\sigma/Y^*_n) = \sigma^2 Var(1/Y^*_n) =  \sigma^2 \frac{\alpha}{(\alpha -1)^2(\alpha -2)}.
			\end{eqnarray}
		\end{proof}
		
		If the sequence of random variables $\{T^*_n\}$ is as defined in equation $(3.10)$ then the auto-correlation function will be 
		\begin{eqnarray}
			Corr(T^*_0, T^*_1) = \frac{E(T^*_0 T^*_1) - \big(\frac{\sigma \alpha}{\alpha -1}\big)^2}{\frac{\sigma^2\alpha}{(\alpha -1)^2(\alpha -2)}}.
		\end{eqnarray}
		which is monotone in $\delta$.
		
		The only thing needed in order to get the auto-correlation is $E(T^*_0 T^*_1)$. Now,
		%==================Edited========22feb======1pm======
		\begin{eqnarray}
			E(T^*_0 T^*_1) &=& E((\sigma/Y^*_0)(\sigma/Y^*_1) ) = \sigma^2 E((1/Y^*_0) (1/Y^*_1)) \nonumber \\
			&=& \sigma^2 E \big(  U_0^{-1/\alpha}   \min \{ (T_0^*)^{-\frac{\alpha}{\alpha-\delta}}, U^{-1/\delta}_1\} \big) \nonumber\\
			&=& \sigma^2 E \big(  U_0^{-1/\alpha}   \min \{ (U_0)^{-\frac{1}{\alpha-\delta}}, U^{-1/\delta}_1\} \big) \nonumber\\
			&=& \sigma^2 \int_{0}^{1} u^{-\frac{1}{\alpha}}_0 \Bigg[ \int_{0}^{u^{\frac{\delta}{\alpha-\delta}}_0}  u^{-\frac{1}{\alpha -\delta}}_1 du_1\Bigg] du_0  +  \sigma^2 \int_{0}^{1} u^{-\frac{1}{\alpha}}_0 \Bigg[ \int_{u^{\frac{\delta}{\alpha-\delta}}_0}^{1}  u^{-\frac{1}{ \delta}}_0 du_1\Bigg] du_0 \nonumber  \\
			&=& \sigma^2 \frac{\alpha \delta}{\delta - 1} \Bigg[ \frac{1}{ \alpha -1} - \frac{\alpha-\delta }{\alpha^2 - 2\alpha+\delta}\Bigg] + \sigma^2 \frac{\alpha(\alpha-\delta)}{\alpha^{2}-2\alpha+\delta}.\nonumber
		\end{eqnarray}  
		%=================Upto here =Edited========22feb======1pm======
		
		\begin{theorem}
			If the sequence of random variables $\{T^*_n\}$ is as defined in equation $(3.10)$ then 
			\begin{eqnarray}
				P(T^*_1< T^*_0) = \frac{\delta}{\alpha + \delta}.
			\end{eqnarray}	
		\end{theorem}
		\begin{proof}
			By the definition of the process we have $T_1^*= \sigma/Y_1^*$, consequently $P(T^*_1< T^*_0)$ is given by
			\begin{eqnarray}
				P(T^*_1 < T^*_0) = P(\sigma/Y^*_1< \sigma/Y^*_0)= P(Y^*_1 > Y^*_0) = \frac{\delta}{\alpha +\delta}
			\end{eqnarray}
			For the computation of $P(Y^*_1  >Y^*_0)$, we refer to page 19 of Arnold et al. \cite{asm23}.
		\end{proof}

		%%%% Inference

		\section{Statistical Inference}
		\subsection{Estimation}
		In the following, we use the method of moments to obtain consistent estimators for the parameters for both of the CPFD processes.
		
		\subsubsection{Method of moments}
		We define the following statistics for the observed sample paths $\{V^*_1,..., V^*_m\}$ from the Kundu CPFD process given in $(3.3)$.
		\begin{eqnarray}
			\bar{V^*}  &=& \frac{1}{m}\sum_{i=1}^{m} V^*_i  \\
			P &=& \frac{1}{m} \sum_{i=1}^{m}  I(V^*_i <V^*_{i-1}).
		\end{eqnarray}
		where $I(.)$ is the indicator function. \\

		For the Kundu CPFD process in $(3.3)$ and the expression for $P(V^*_1 <V^*_2)$ for the case $\alpha > \beta$, the following two moment equations can be solved for estimates of $\alpha$ and $\beta$.
		\begin{eqnarray}
			\overline{V^*}&=& \frac{1}{\alpha + \beta +1}
			\\
			P&=&\frac{\alpha + \beta}{2\alpha+\beta}.
		\end{eqnarray}
		
		We obtain the following consistent estimates.
		
		\begin{eqnarray}
			\hat{\alpha}&=& \frac{1-\overline{V^*} (1+ 	\hat{\beta})}{\overline{V^*}}
			\\
			\hat{\beta}&=& \frac{(2P-1)(1-	\overline{V^*})}{\overline{V^*}P}.
		\end{eqnarray}
		
		This solution does not always satisfy the condition $\alpha > \beta$. Now, for the case $\alpha < \beta$  we have
		
		$$P(V^*_1>V^*_2)=\frac{\alpha}{2\alpha+\beta}.$$
		
		The following two moment equations can be solved for estimates of $\alpha$ and $\beta$.
		\begin{eqnarray}
			\overline{V^*}&=& \frac{1}{\alpha + \beta +1}
			\\
			P&=&\frac{\beta}{2\beta+\alpha}.
		\end{eqnarray}
		
		We can then obtain the following consistent estimates.
		
		\begin{eqnarray}
			\hat{\beta}&=& \frac{(1-\overline{V^*})P}{(1-P)\overline{V^*}}
			\\
			\hat{\alpha}&=& \frac{\hat{\beta}(1-2P)}{P}
		\end{eqnarray}
		
		Now, for the A-M  CPFD process in $(3.4)$ and the expression for $P(W^*_1 <W^*_2)$, the following two moment equations can be solved for estimates of $\alpha$ and $\beta$.
		\begin{eqnarray}
			\overline{W^*}&=& \frac{1}{\alpha +1}
			\\
			P&=&\frac{\delta}{\alpha+\delta}.
		\end{eqnarray}
		
		We obtain the following consistent estimates.
		
		\begin{eqnarray}
			\hat{\alpha}&=& \frac{1-\overline{W^*}}{\overline{W^*}}
			\\
			\hat{\delta}&=& \frac{P(1- \overline{W^*})}{\overline{W^*}(1-P)}.
		\end{eqnarray}
		
		For the Kundu Pareto process in $(3.9)$ and the expression for $P(S^*_1<S^*_2)$, for the case, $\alpha>\beta$, the following two moment equations can be solved for estimates of $\alpha$ and $\beta$. We also know that a consistent estimator  for $\sigma $  is	$$\hat{\sigma}= S_{(1)}= \min\{S_1,...,S_m\}.$$
		Now, using the following two moment equations
		\begin{eqnarray}
			\overline{S^*}&=& \sigma  \frac{\alpha + \beta }{\alpha +\beta-1}
			\\
			P&=&\frac{\alpha +\beta}{2\alpha+\beta}.
		\end{eqnarray}

		We obtain the following consistent estimates.
		
		\begin{eqnarray}
			\hat{\alpha}&=& \frac{(1-P)\overline{S^*}}{P(\overline{S^*}-\hat{\sigma})}
			\\
			\hat{\beta}&=&  \frac{\hat{\alpha}(2P-1)}{1-P}.
		\end{eqnarray}
		
		Now, for the case $\alpha <\beta$, the two moment equations are 	
		
		\begin{eqnarray}
			\overline{S^*}&=& \sigma  \frac{\alpha + \beta }{\alpha +\beta-1}
			\\
			P&=&\frac{\beta}{2\beta+\alpha}.
		\end{eqnarray}

		We obtain the following consistent estimates.
		
		\begin{eqnarray}
			\hat{\alpha}&=& \frac{\overline{S^*}(1-2P)}{(1-P)(\overline{S^*}-\hat{\sigma})}
			\\
			\hat{\beta}&=&\frac{\hat{\alpha} P}{(1-2P)}.
		\end{eqnarray}

		For an A-M Pareto process in $(3.10)$ and the expression for $P(T^*_1<T^*_2)$, the following two moment equations can be solved for estimates of $\alpha$ and $\delta$. We also know that consistent estimator  for $\sigma $  is	$$\hat{\sigma}= S_{(1)}= \min\{S_1,...,S_m.\}$$ Now,

		\begin{eqnarray}
			\overline{T^*}&=& \sigma \frac{\alpha}{\alpha-1}
			\\
			P&=& \frac{\delta}{\alpha + \delta}.
		\end{eqnarray}

		We obtain the following consistent estimates.
		
		\begin{eqnarray}
			\hat{\delta}&=& \frac{\overline{T^*}P}{(\overline{T^*}-\hat{\sigma})(1-P)}
			\\
			\hat{\alpha}&=&  \frac{	\hat{\delta}(1-P)}{P}.
		\end{eqnarray}

		\section{Applications}
		In the following two subsections, we illustrate a simulation study and give examples
		of real-life applications of the two CPFD-processes given in $(3.3)$ and $(3.4)$  and two Pareto processes given in $(3.7)$ and $(3.8)$.

		\subsection{Simulation study}
		For the processes given in equation  $(3.3)$ \& $(3.4)$ and $(3.7)$\& $(3.8)$ for the same moment estimators given in Section 5, we illustrate a simulation procedure on the behavior of estimators by varying the parameter values with increasing sample path sizes.  
		
		We have simulated $2000$ data sets of sample path size  $m = 20, 30, 50, 100, 200, 500$ for the parameter  vectors $(\alpha=0.5,\beta=0.1)$ for the process in  $(3.3)$ and $(\alpha=1,\delta=0.1)$ for the process in  $(3.4)$ .
		Similarly,  for the same number of data sets with the same varying sample path sizes, we have simulated observations from the process in $(3.7)$ for the parameter vectors  
		$(\sigma=1, \alpha=1, \beta =2)$ and  $(\sigma=1, \alpha=4, \delta =2)$ from the process in  $(3.8)$ .

		We refer to the following Figures \ref{fig3}--\ref{fig6} for the bootstrapped distribution of each parameter estimate for the processes defined above. The numerical evidence suggests that as sample size increases, the moment estimates approach the true parameter values with standard errors that decrease as the samlpe size increases.

		\begin{figure} [h!]
			\begin{subfigure}{8cm}
				\centering
				\includegraphics[width=8cm]{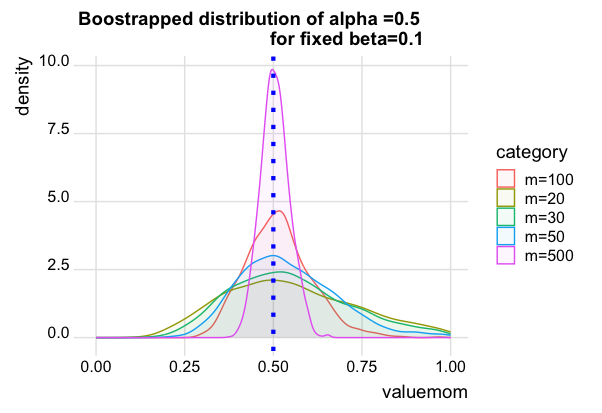}
				\caption{$\alpha$} 
			\end{subfigure}
			\begin{subfigure}{8cm}
				\centering
				\includegraphics[width=8cm]{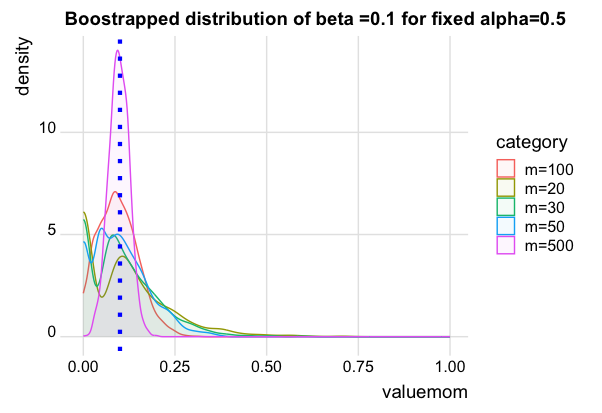}
				\caption{$\beta$} 
			\end{subfigure}
			\caption{Kundu CPFD process }
			\label{fig3}
		\end{figure}
		
		\begin{figure} [h!]
			\begin{subfigure}{8cm}
				\centering
				\includegraphics[width=8cm]{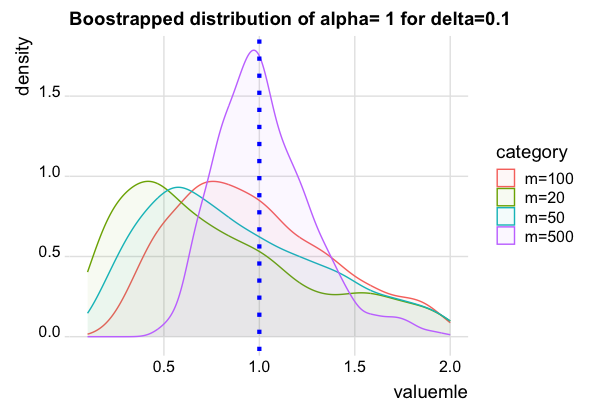}
				\caption{$\alpha$} 
			\end{subfigure}
			\begin{subfigure}{8cm}
				\centering
				\includegraphics[width=8cm]{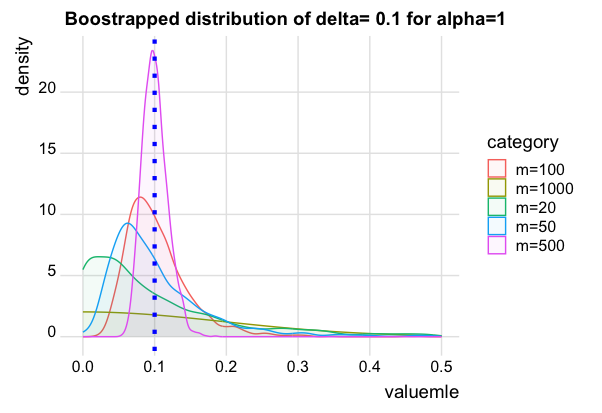}
				\caption{$\delta$} 
			\end{subfigure}
			\caption{A-M CPFD process }
			\label{fig4}
		\end{figure}

		\begin{figure} [h!]
			\begin{subfigure}{5cm}
				\centering
				\includegraphics[width=5cm]{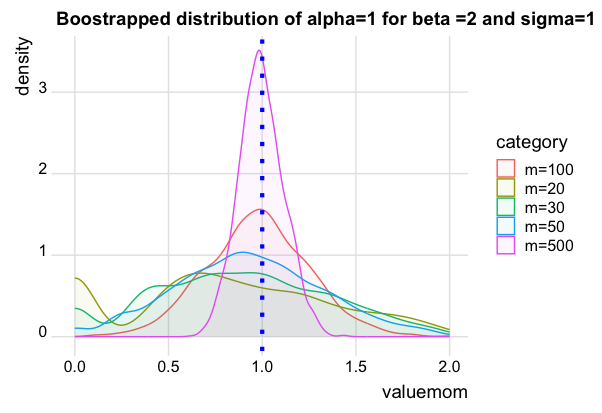}
				\caption{$\alpha$} 
			\end{subfigure}
			\begin{subfigure}{5cm}
				\centering
				\includegraphics[width=5cm]{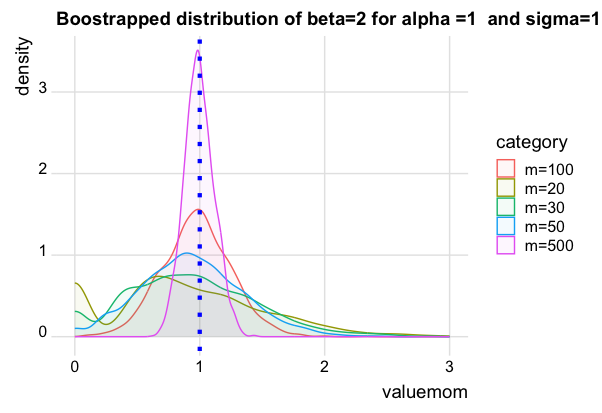}
				\caption{$\beta$} 
			\end{subfigure}
			\begin{subfigure}{5cm}
				\centering
				\includegraphics[width=5 cm]{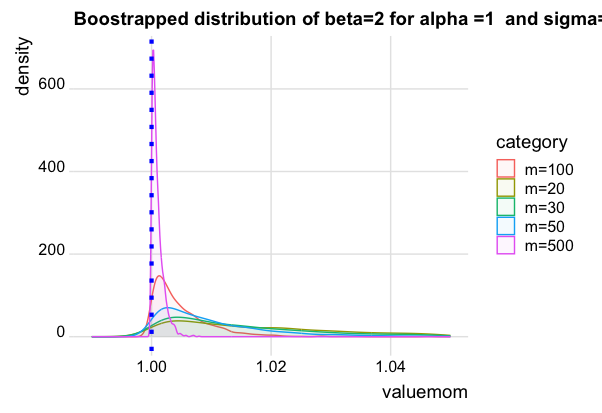}
				\caption{$\sigma$} 
			\end{subfigure}
			\caption{Kundu Pareto Process }
			\label{fig5}
		\end{figure}

		\begin{figure} [h!]
			\begin{subfigure}{5cm}
				\centering
				\includegraphics[width=5cm]{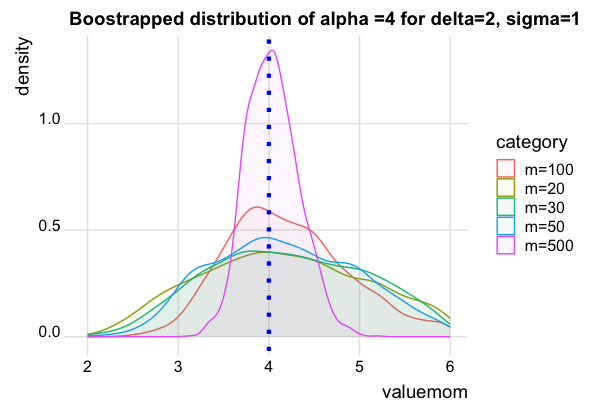}
				\caption{$\alpha$} 
			\end{subfigure}
			\begin{subfigure}{5cm}
				\centering
				\includegraphics[width=5cm]{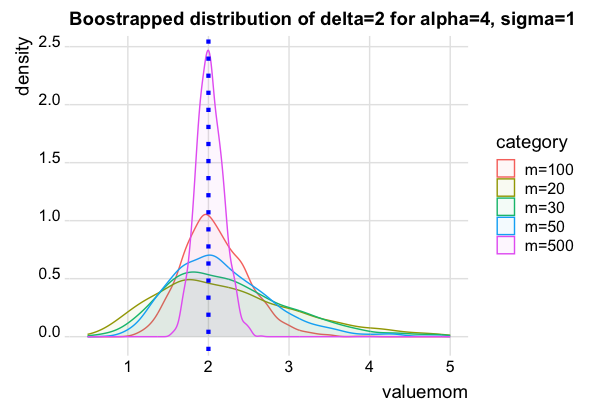}
				\caption{$\delta$} 
			\end{subfigure}
			\begin{subfigure}{5cm}
				\centering
				\includegraphics[width=5cm]{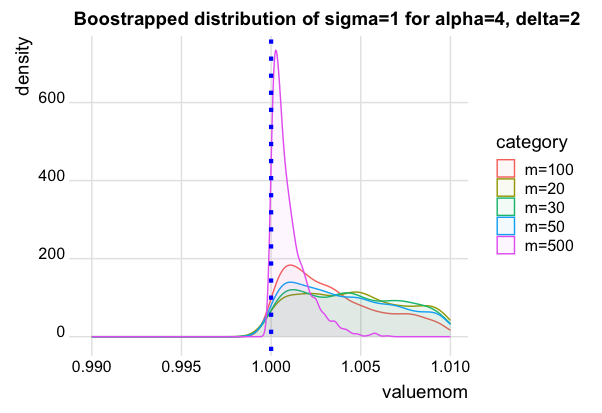}
				\caption{$\sigma$} 
			\end{subfigure}
			\caption{AM Pareto Process}
			\label{fig6}
		\end{figure}

		\subsection{Real-life data}
		
	In the following, we consider S\&P data set from 3 January 1928 to 20 February 1928 of size $34$ observations. The data set can be obtained from the website  \href{https://www.kaggle.com/datasets/camnugent/sandp500}{\nolinkurl{https://www.kaggle.com/datasets/camnugent/sandp500}}.  
	
	Here, we assume that the data generating process for the S\&P data is a proportional hazard process with unknown marginal distribution $F^*$ and underlying CPFD process of the kind in either  $(3.3)$ or $(3.4)$. We used a non-parametric approach to estimate $1-F^*$ by using  $1-F_m$, the complementary empirical distribution function of the $m=34$ available data points of S\&P. Now, the transformation $1-F_m$ applied to the available  S\&P data will be approximately a sample from a CPFD process.

	We refer to Table \ref{table:1} and \ref{table:2} for the fitted parameter values for the two CPFD processes and the mean square error(MSE)  between the complementary empirical distribution function with the estimated CPFD processes. We can conclude that the AM CPFD Process $(3.4)$ fit reasonably well for the data.
	
	For the new Pareto processes, we make similar assumptions but without applying a  distributional transformation to the available data.  We refer to Table  \ref{table:3} and \ref{table:4} for the fitted parameter values for the two new Pareto processes and the mean square error(MSE)  between the actual  S\&P data with the estimated Pareto processes. We may conclude that a Kundu Pareto process ( 3.7 ) provides a reasonable fit for the data.
	
	Unfortunately, there is no proper goodness of fit measure available to specifically quantify how well the processes fit the data.  The problem of assessing the suitability of either of the power function process models to a real-world data set will require further investigation.

		 \begin{table}[h!]
		 	\centering
		 	\begin{tabular}{||c | c ||} 
		 		\hline
		 		Parameters & Estimates \\ [0.5ex] 
		 		\hline\hline
		 		$\alpha$ & $0.963$ \\ 	\hline
		 		$\beta$ & $0.1203$ \\ 	\hline
		 		MSE & $0.212$\\  [1ex] 
		 		\hline
		 	\end{tabular}
		 	\caption{S\&P data fitted to $(3.3)$ process}
		 	\label{table:1}
		 \end{table}
		 
		 \begin{table}[h!]
		 	\centering
		 	\begin{tabular}{||c | c ||} 
		 		\hline
		 		Parameters & Estimates \\ [0.5ex] 
		 		\hline\hline
		 		$\alpha$ & $1.083$ \\ 	\hline
		 		$\delta$ &  $0.855$\\ 	\hline
		 		MSE & $0.155$ \\  [1ex] 
		 		\hline
		 	\end{tabular}
		 	\caption{S\&P data fitted to $(3.4)$ process}
		 	\label{table:2}
		 \end{table}

			 \begin{table}[h!]
			\centering
			\begin{tabular}{||c | c ||} 
				\hline
				Parameters & Estimates \\ [0.5ex] 
				\hline\hline
				$\sigma$ & $16.950$ \\ 	\hline
				$\alpha$ & $7.015$ \\ 	\hline
				$\beta$ & $26.307$ \\ 	\hline
				MSE & $0.257$\\  [1ex] 
				\hline
			\end{tabular}
			\caption{S\&P data fitted to $(3.7)$ process}
			\label{table:3}
		\end{table}
		
		\begin{table}[h!]
			\centering
			\begin{tabular}{||c | c ||} 
				\hline
				Parameters & Estimates \\ [0.5ex] 
				\hline\hline
				$\sigma$ & $16.950$ \\ 	\hline
				$\alpha$ & $33.321$ \\ 	\hline
				$\delta$ &  $37.487$\\ 	\hline
				MSE & $9.636$ \\  [1ex] 
				\hline
			\end{tabular}
			\caption{S\&P data fitted to $(3.8)$process}
			\label{table:4}
		\end{table}

		\section{Conclusion}
		
		We have considered two complementary power function distribution processes and two new Pareto processes, all stationary while some are Markovian. We have derived the marginal and bivariate distributional properties and included moment computation up to the auto-correlation function of all of the processes. Since the derived bivariate distributions do not have densities, likelihood-based inferences such as maximum likelihood and the Bayesian approach cannot be used to estimate the process parameters. However, we have derived one more variation property of each of the processes and have used it to obtain simple expressions for consistent estimators of the parameters of the processes. We also included a simulation study of the behavior of the estimators with varying sample  sizes.   Finally, we included an example of an application to the S\&P data and illustrated the general procedure to fit such processes.  The investigation of such processes is still at the beginning and will eventually merit a place in the toolkit of modelers dealing with proportional hazard processes.
		
		\section*{Acknowledgment}
		The second author, S. Sachdeva's research, is Sponsered by the DST INSPIRE.
		
		\section*{Conflicts of Interest}
		The authors declare no conflict of interest.

	\end{document}